\documentclass[default]{sn-jnl}

\jyear{2022}%

\theoremstyle{thmstyleone}%
\newtheorem{theorem}{Theorem}
\newtheorem{proposition}[theorem]{Proposition}%

\theoremstyle{thmstyletwo}%

\theoremstyle{thmstylethree}%

\newcommand{\be}{\begin{equation}}
\newcommand{\ee}{\end{equation}}

\numberwithin{equation}{section}

\begin{document}

\title[On the exactly-solvable semi-infinite quantum well ...]{On the exactly-solvable semi-infinite quantum well of the non-rectangular step-harmonic profile}


\author*[1]{\fnm{E.I.} \sur{Jafarov}}\email{ejafarov@physics.science.az}

\author[1]{\fnm{S.M.} \sur{Nagiyev}}\email{smnagiyev@physics.ab.az}

\affil*[1]{\orgdiv{Institute of Physics}, \orgname{Azerbaijan National Academy of Sciences}, \orgaddress{\street{Javid av. 131}, \city{Baku}, \postcode{AZ1143}, \country{Azerbaijan}}}


\abstract{An exactly-solvable model of the non-relativistic harmonic oscillator with a position-dependent effective mass is constructed. The model behaves itself as a semi-infinite quantum well of the non-rectangular profile. Such a form of the profile looks like a step-harmonic potential as a consequence of the certain analytical dependence of the effective mass from the position and semiconfinement parameter $a$. Both states of the discrete and continuous spectrum are studied. In the case of the discrete spectrum, wavefunctions of the oscillator model are expressed through the Bessel polynomials. The discrete energy spectrum is non-equidistant and finite as a consequence of its dependence on parameter $a$, too. In the case of the continuous spectrum, wavefunctions of the oscillator model are expressed through the $_1F_1$ hypergeometric functions. At the limit, when the parameter $a$ goes to infinity, both wavefunctions, and the energy spectrum of the model under construction correctly reduce to corresponding results of the usual non-relativistic harmonic oscillator with a constant effective mass. Namely, wavefunctions of the discrete spectrum recover wavefunctions in terms of the Hermite polynomials, and wavefunctions of the continuous spectrum simply vanish. We also present a new limit relation that reduces Bessel polynomials directly to Hermite polynomials and prove its correctness using the mathematical induction technique.}

\keywords{Semi-infinite potential, Exact solution, Harmonic oscillator, Position-dependent effective mass, Bessel polynomials}



\maketitle

\section{Introduction}

Low-dimensional quantum systems, especially quantum well systems are of great interest in both theoretical and experimental physics as well as related areas due to the recent impact of such quantum systems on the development of modern nanotechnologies and advanced devices. Many physical phenomena related to modern nanotechnologies can be successfully explained through the exact or numerical solution of one-dimensional Schr\"odinger equation with effective mass $m_0$ and time-independent potential $V\left(x\right)$. Quantum well systems are among them, too. Most of the quantum well potentials leading to the exact solution of the corresponding Schr\"odinger equation are symmetric in respect to position inversion due to that, such symmetric behavior of the potential substantially simplifies the exact solution of the Schr\"odinger equation for it. However, the recent development of the molecular-beam epitaxy and related pioneering methods on the experimental fabrication of the artificial quantum well structures with profiles differing from the symmetric ones~\cite{miller1984a,miller1984b,miller1985,gossard1986} also requires to describe accurately a number of experimental results in the nanotechnologies via the exactly-solvable potentials behaving themselves as confined quantum well systems with semi-infinite property and non-rectangular profile.

In the present paper, we report on the new quantum oscillator model that behaves itself as a semi-infinite quantum well of the non-rectangular profile. If to be exact, then this quantum system is confined by an infinitely high wall at an arbitrary negative finite value of a position and with a finite wall at an arbitrary positive finite value of the position. Also, the system behaves itself as a non-relativistic quantum harmonic oscillator within these two values of the position, therefore, it has a completely non-rectangular profile. In other words, its behavior almost overlaps with the confined quantum system under the so-called step-harmonic potential~\cite{rizzi2010,amthong2014}, but differs from it that the model under the current study does not vary abruptly. Such a behavior of the quantum system is achieved thanks to the introduction of the effective mass varying with position~\cite{morris2015,morris2017}.

We structured our paper as follows: the next section presents a brief review of the exact solution of the non-relativistic quantum harmonic oscillator within the canonical approach. Then, Sect. 3 consists of the main result of the current paper -- it is mainly devoted to the construction of the exactly solvable model of the non-relativistic harmonic oscillator with a position-dependent effective mass, which behaves itself as a semi-infinite quantum well of the non-rectangular profile. In the final section, we discuss the limit cases, when the parameter $a$ goes to infinity and as a consequence, the wavefunctions of the model under construction expressed by the Bessel polynomials completely recover the wavefunctions of the non-relativistic quantum harmonic oscillator within the canonical approach as well as energy spectrum of non-equidistant and finite form becomes equidistant and infinite.

\section{The non-relativistic quantum harmonic oscillator within the canonical approach}

In this informative section, we are going to provide briefly mathematical expressions belonging to the non-relativistic quantum harmonic oscillator within the canonical approach. These expressions also can be found easily in most quantum mechanics textbooks. But, we decided to include them in our paper, because, they will be useful during the computations being done in the next section. 

Wavefunctions of the stationary states and the energy spectrum of the non-relativistic quantum harmonic oscillator within the canonical approach can be obtained exactly from the one-dimensional time-independent Schr\"odinger equation

\be
\label{sheq-gen}
\hat H\psi \left( x \right) = E\psi \left( x \right),
\ee
where the full Hamiltonian is a sum of the kinetic and potential energy operators

\be
\label{h}
\hat H = \hat H_0  + V\left( x \right),
\ee
with the kinetic energy operator of the following general form:

\be
\label{h0}
\hat H_0  = \frac{{\hat p_x ^2 }}{{2m_0 }}.
\ee

Here, the momentum operator $\hat p_x$ being defined within the canonical approach is known as

\be
\label{p_x}
\hat p_x  =  - i\hbar \frac{\mathrm{d}}{{\mathrm{d}x}}.
\ee

We are going to deal with the model based on the harmonic oscillator, therefore, its potential $V\left( x \right)$ with states bounded at infinity is defined as

\be
\label{v-x}
V\left( x \right) = \frac{{m_0 \omega ^2 x^2 }}{2},\quad  - \infty  < x <  + \infty .
\ee

Here, $\omega$ is its constant angular frequency. Substitution of (\ref{h0})--(\ref{v-x}) at (\ref{h}) leads to the following second order differential equation generated from Eq.(\ref{sheq-gen}):

\be
\label{sheq-osc}
\frac{{\hbar ^2 }}{{2m_0 }}\frac{{\mathrm{d}^2 \psi }}{{\mathrm{d}x^2 }} + \left( {E - \frac{{m_0 \omega ^2 x^2 }}{2}} \right)\psi  = 0.
\ee

Exact solution of this equation allows to obtain equidistant energy spectrum $E \equiv E_n$ and wavefunctions of the stationary states $\psi \left( x \right) \equiv \psi_n \left( x \right)$ of the following analytical form:

\be
\label{en}
E_n  = \hbar \omega \left( {n + \frac{1}{2}} \right),\quad n = 0,1,2, \ldots ,
\ee

\be
\label{wf-on}
\psi _n \left( x \right) = \frac{1}{{\sqrt {2^n n!} }}\left( {\frac{{m_0\omega }}{{\pi \hbar }}} \right)^{{\textstyle{1 \over 4}}} \mathrm{e}^{ - \frac{{m_0\omega x^2 }}{{2\hbar }}} H_n \left( {\sqrt {\frac{{m_0\omega }}{\hbar }} x} \right).
\ee

Here, $H_n \left( { x} \right)$ are Hermite polynomials defined in terms of the $_2F_0$ hypergeometric functions as follows~\cite{koekoek2010}:

\be
\label{hermite}
H_n (x) = (2x)^n \,_2 F_0 \left( {\begin{array}{*{20}c}
   {\begin{array}{*{20}c}
   { - n/2, - (n - 1)/2}  \\
    -   \\
\end{array};} & { - \frac{1}{{x^2 }}}  \\
\end{array}} \right).
\ee

These polynomials also satisfy the following recurrence relations, which we will use in the final section of the current paper:

\be
\label{rr-h}
H_{n + 1} \left( x \right) = 2xH_n \left( x \right) - 2nH_{n - 1} \left( x \right).
\ee

$\psi _n \left( x \right)$ wavefunctions (\ref{wf-on}) are orthonormalized and they satisfy the following orthogonality relation:

\be
\label{wf-ort}
\int\limits_{ - \infty }^\infty  {{\psi _m ^* (x)}{\psi} _n (x)dx}  = \delta _{mn} .
\ee

Hamiltonian (\ref{h}) being sum of the non-relativistic canonical kinetic (\ref{h0}) and harmonic potential energy operators (\ref{v-x}) can be easily factorized in terms of the harmonic oscillator creation and annihilation operators as follows~\cite{dirac1927,infeld1951,messiah1966}:

\be
\label{h-f}
\hat H = \hbar \omega \left( {\hat a^ +  \hat a^- + \frac{1}{2}} \right).
\ee

Here, harmonic oscillator creation and annihilation operators have the following first-order differential analytical expression:

\begin{eqnarray}
\label{aa+}
 \hat a^ +   = \frac{1}{{\sqrt 2 \lambda _0 }}\left( {\lambda _0 ^2 x - \frac{\mathrm{d}}{{\mathrm{d}x}}} \right), \\ 
 \hat a^- = \frac{1}{{\sqrt 2 \lambda _0 }}\left( {\lambda _0 ^2 x + \frac{\mathrm{d}}{{\mathrm{d}x}}} \right),  \nonumber
\end{eqnarray}
where, $\lambda _0  = \sqrt {\frac{{m_0 \omega }}{\hbar }}$.

We note here two main properties of these operators, one of which is that they do not commute, therefore their commutation is one, i.e.

\[
\left[ {\hat a^-,\hat a^ +  } \right] = 1,
\]
and the second one is the following equation, where the action of the annihilation operator $\hat a^-$ to the ground state of wavefunction (\ref{wf-on}) is zero:

\[
\hat a^- \psi _0 \left( x \right) = 0.
\]

\section{The model of the semi-infinite quantum well of the non-rectangular profile}

In this section, we are going to construct the exactly solvable model of the non-relativistic harmonic oscillator with a position-dependent effective mass, which behaves itself as a semi-infinite quantum well of the non-rectangular profile. In previous section, it was highlighted that the non-relativistic harmonic oscillator with a potential defined via (\ref{v-x}) has the wavefunctions (\ref{wf-on}), which vanish at $\pm\infty$. Then, we decided that if we want to make from the non-relativistic harmonic oscillator a semi-infinite quantum well of the non-rectangular profile in the form of the step-harmonic potential, here the effective mass certainly depending on the position can be a powerful tool. Before showing a certain implementation of this idea, we want to share some known information about the conception of the position-dependent effective mass in quantum mechanics. This conception is presented in~\cite{bendaniel1966}, which supposed that the change of the band structure of the independent-particle model of the superconductor~\cite{harrison1961} on the tunneling experiment~\cite{giaever1960a,giaever1960b} should be simulated by a spatially varying effective mass $M \left(x\right)$. In other words, it is necessary to implement a replacement $m_0 \to M \left(x\right)$ in full Hamiltonian~(\ref{h}) and then to introduce a new analogue of the kinetic energy operator~(\ref{h0}). The following simpler realization of the Hermitian kinetic energy operator, now called as BenDaniel-Duke kinetic energy operator has been introduced:

\begin{equation}
\label{h0-pdem}
\hat H_0 \equiv \hat H_0^{BD} =  - \frac{{\hbar ^2 }}{2}\frac{\mathrm{d}}{{\mathrm{d}x}}\frac{1}{{M (x)}}\frac{\mathrm{d}}{{\mathrm{d}x}}.
\end{equation}

\begin{figure}[t!]
\begin{center}
\resizebox{0.35\textwidth}{!}{%
  \includegraphics{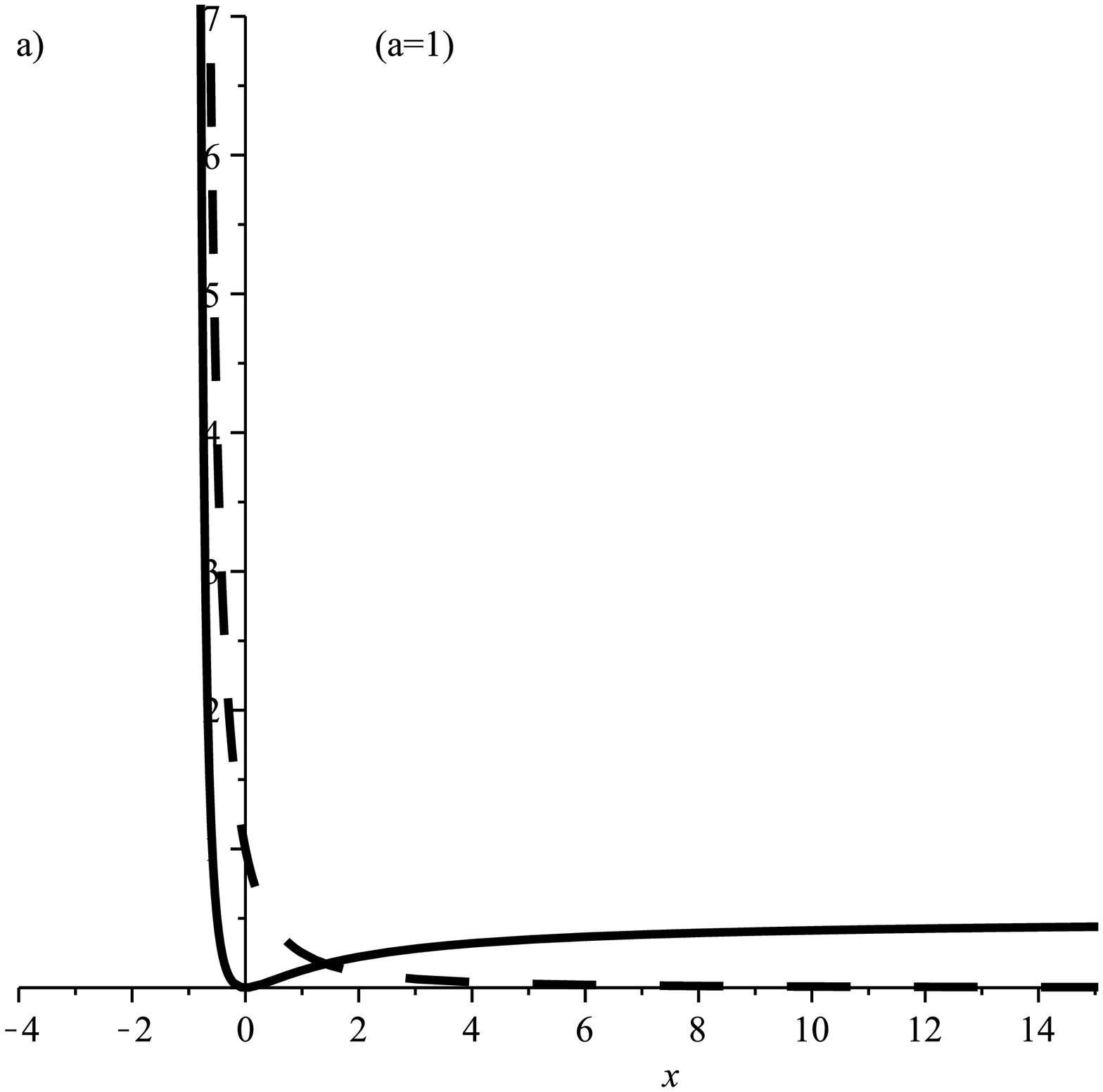}
}
\resizebox{0.35\textwidth}{!}{%
  \includegraphics{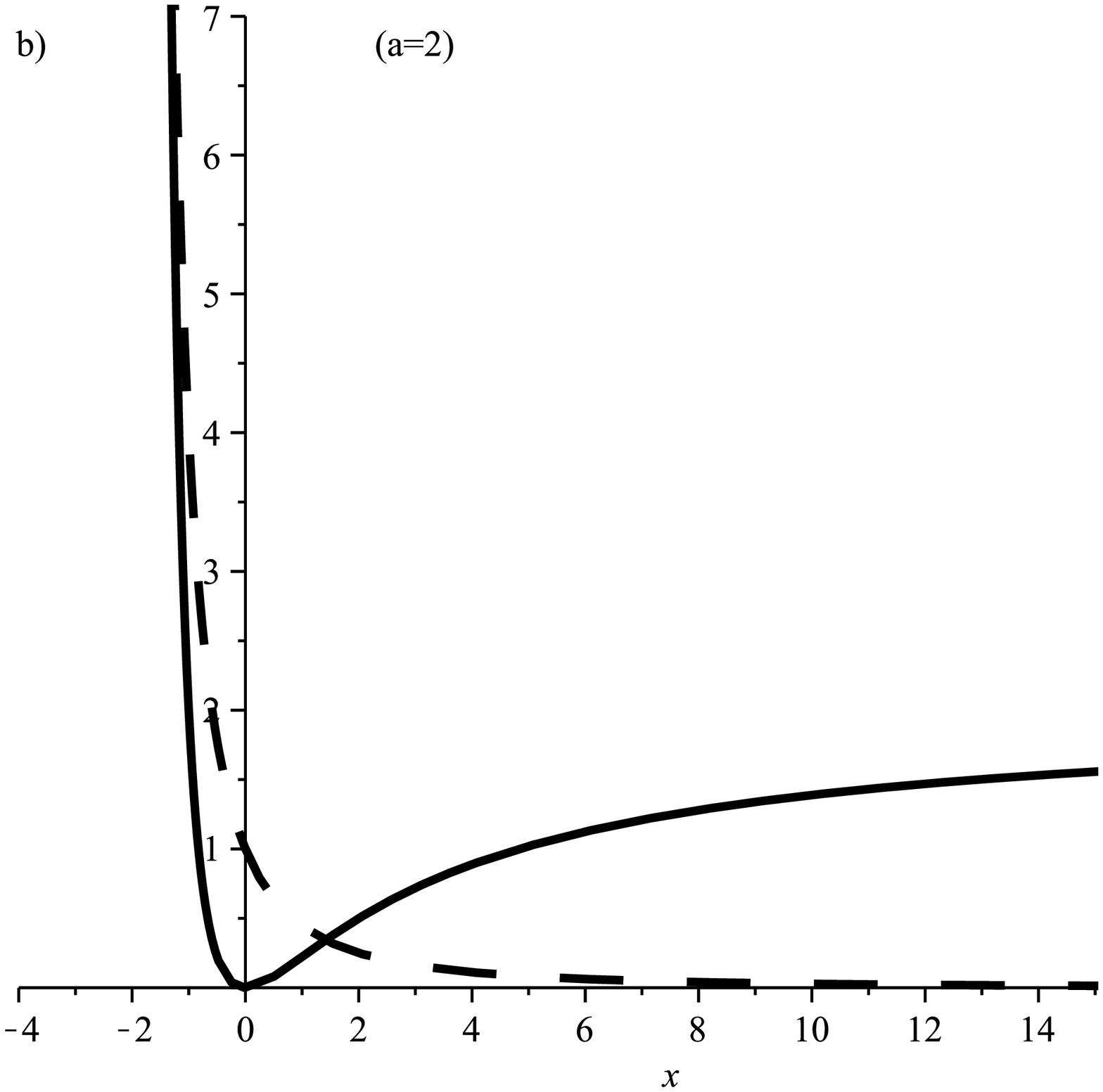}
}\\
\resizebox{0.35\textwidth}{!}{%
  \includegraphics{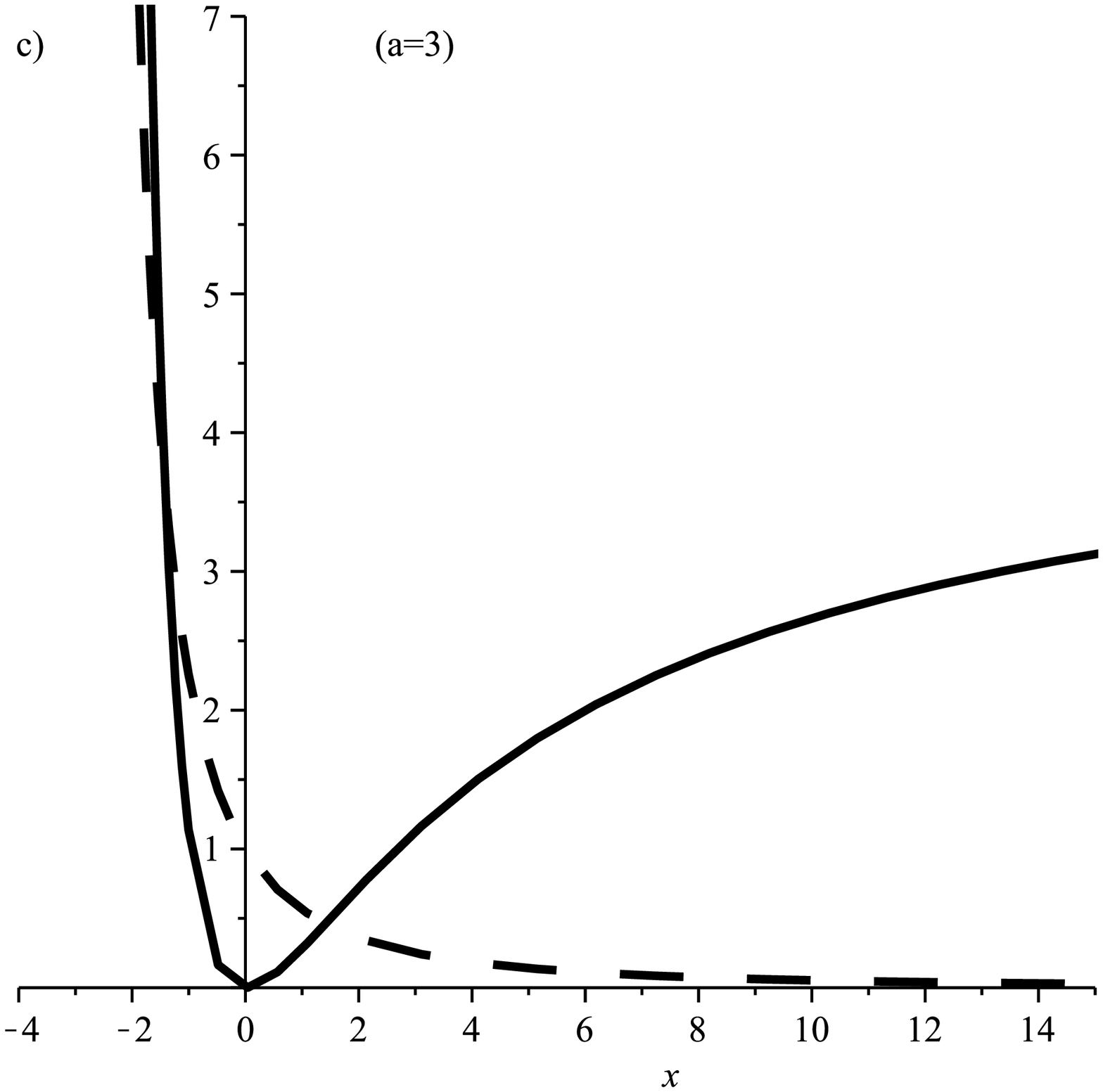}
}
\resizebox{0.35\textwidth}{!}{%
  \includegraphics{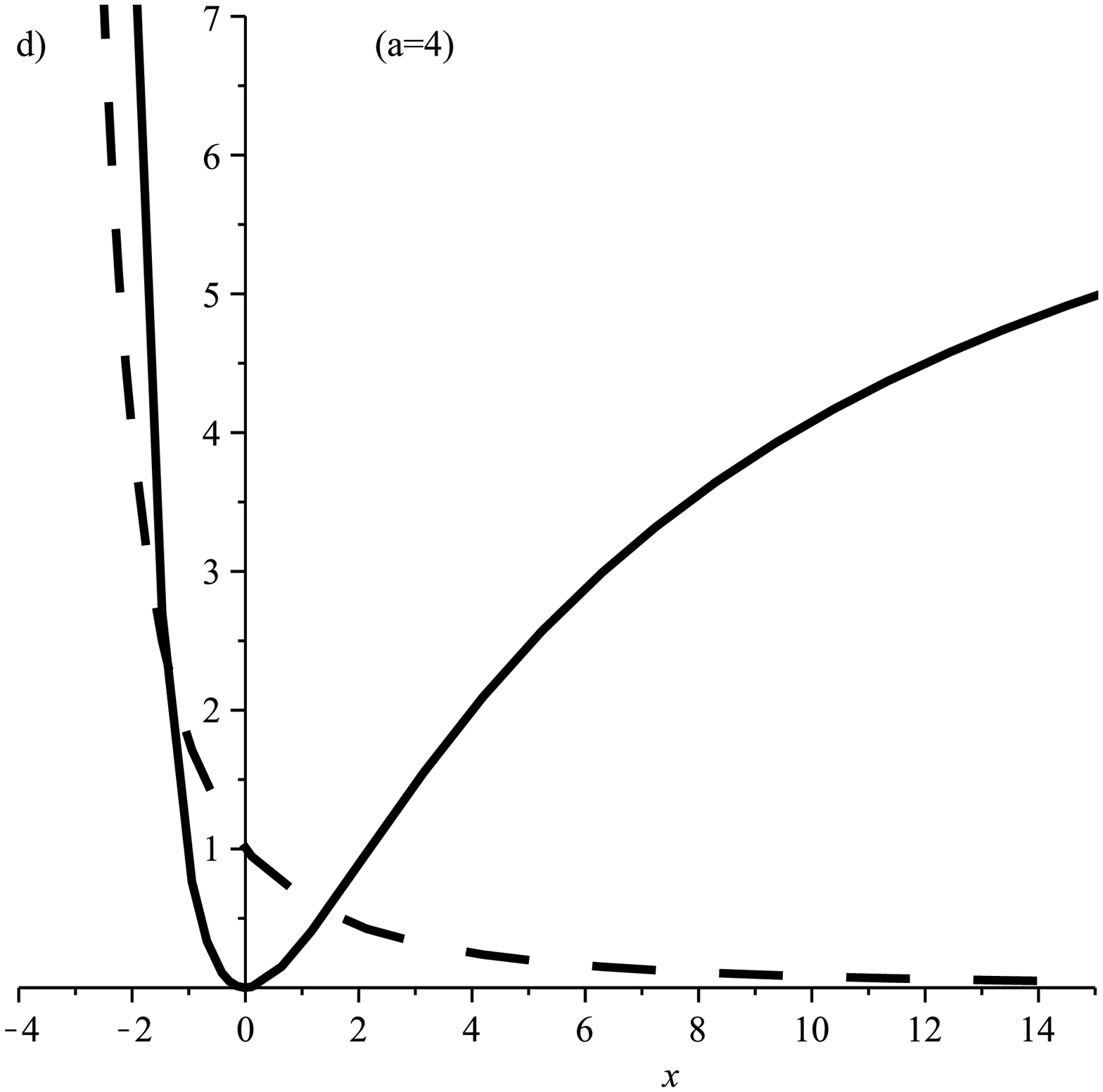}
}
\end{center}
\caption{
Behavior of the semi-infinite quantum well potential~(\ref{vx}) of the non-rectangular step-harmonic profile (solid line) and the effective mass $M\left( x \right)$ varying by position~(\ref{m-x}) (dashed line) for the confinement parameter: a) $a=1$; b) $a=2$; c) $a=3$; d) $a=4$ ($m_0=\omega=\hbar=1$, hence, $\lambda _0 = 1$).} 
\label{fig.1}
\end{figure}

Despite that there is a number of other definitions for the position-dependent effective mass kinetic energy operator~\cite{gora1969,zhu1983,vonroos1983,morrow1984,li1993,li2009,lima2012,nobre2015,mustafa2019,elnabulsi2020}, we are going to use in our computations BenDaniel--Duke definition~(\ref{h0-pdem}) due to that from mathematics viewpoint, the Schr\"odinger equation with the other kinetic energy operator definitions does not differ significantly from BenDaniel--Duke one and the method used here can be later easily applied for them, too. Also, it is shown in~\cite{kolesnikov1999} that if to modify BenDaniel--Duke kinetic energy operator (\ref{h0-pdem}), then additional singular terms will appear in its modified version and they can lead to a discontinuity of the wavefunction.

By substitution of (\ref{h0-pdem}) at (\ref{h}) one has the following Schr\"odinger equation:

\be
\label{seq-pdem1}
\frac{{\hbar ^2 }}{2}\frac{\mathrm{d}}{{\mathrm{d}x}}\frac{1}{M}\frac{\mathrm{d}}{{\mathrm{d}x}}\psi  + \left( {E - V} \right)\psi  = 0.
\ee

Here, by definition $M \equiv M\left(x\right)$ and $V \equiv V\left(x\right)$. As a next step, one needs to introduce a potential $V\left(x\right)$ that will behave itself as a semi-infinite quantum well of the non-rectangular profile. In other words, the quantum system under study will be confined at an arbitrary negative finite value of a position by an infinitely high wall ($x=-a$), but at an arbitrary positive finite value of position with a finite wall and is non-rectangular profile between these two walls will be achieved thanks to its behavior as a non-relativistic quantum harmonic oscillator within these two values of the position. Therefore, it looks like a step-harmonic potential, but with no abrupt change. Analytical expression of such a potential $V\left(x\right)$ can be defined as follows:

\be
\label{vx}
V\left( x \right) = \left\{ \begin{array}{ll}
\frac{{m_0\omega ^2 a^2}}{2}, & \hbox{ for }\quad x \to +\infty ,\\ 
\frac{{M\omega ^2 x^2 }}{2}, & \hbox{ for } \quad x >  - a, \\ 
  + \infty , & \hbox{ for } \quad x \leq  - a ,\\ 
 \end{array} \right.
\qquad (a>0).
\ee

It is clear that the depth of the semi-infinite well is as follows:

\[
V_\infty   \equiv \mathop {\lim }\limits_{x \to  + \infty } V\left( x \right) = \frac{{m_0 \omega ^2 a^2 }}{2}.
\]

Then, the basic principles of quantum mechanics require that the energy spectrum is discrete if $E<V_\infty$ and it is continuous if $E>V_\infty$. From definition of potential (\ref{vx}), it is also clear that the height of a finite wall depends on the parameter $a$ that defines the negative value of the position for which the potential becomes infinite. Such a general behavior can be completely achieved via the following analytical definition of $M \equiv M\left(x\right)$:

\be
\label{m-x}
M\left( x \right) = \frac{{a^2 m_0 }}{{\left( {a + x} \right)^2 }}.
\ee

In fig.1, we present a comparative behavior of the semi-infinite quantum well potential~(\ref{vx}) of the non-rectangular step-harmonic profile and the effective mass $M\left( x \right)$ varying by position~(\ref{m-x}) for different values of the confinement parameter $a$, namely, for $a=1$; $2$; $3$ and $4$. One observes that the behavior of both functions extracted from Eqs.~(\ref{vx}) and (\ref{m-x}) completely proves our main goal. The potential~(\ref{vx}) behaves itself as a semi-infinite quantum well potential of the non-rectangular step-harmonic profile and the position-dependent effective mass defined analytically via Eq.(\ref{m-x}) directly defines its behavior as confined by an infinitely high wall at an arbitrary negative finite value of a position and with a finite wall at an arbitrary positive finite value of the position as well as a non-relativistic quantum harmonic oscillator within these two values of the position. Greater values of the parameter $a$ exhibit evidence of the disappearance of the confinement effect from both negative and positive sides of the quantum-well system.

Substitution of (\ref{vx}) and (\ref{m-x}) at (\ref{seq-pdem1}) leads to the following second-order differential equation:

\be
\label{seq-pdem2}
\frac{{\mathrm{d}^2 \psi }}{{\mathrm{d}x^2 }} + \frac{2}{{a + x}}\frac{{\mathrm{d}\psi }}{{\mathrm{d}x}} - \left( {\frac{{\frac{{m_0 ^2 \omega ^2 a^4 }}{{\hbar ^2 }}x^2 }}{{\left( {a + x} \right)^4 }} - \frac{{\frac{{2m_0 a^2 E}}{{\hbar ^2 }}}}{{\left( {a + x} \right)^2 }}} \right)\psi  = 0.
\ee

One applies the transition to the dimensionless variable $\xi  = \frac{x}{a}$ that slightly changes Eq.(\ref{seq-pdem2}) to the following one:

\be
\label{seq-pdem3}
\frac{{\mathrm{d}^2 \psi }}{{\mathrm{d}\xi ^2 }} + \frac{2}{{1 + \xi }}\frac{{\mathrm{d}\psi }}{{\mathrm{d}\xi }} - \left( {\frac{{\frac{{m_0 ^2 \omega ^2 a^4 }}{{\hbar ^2 }}\xi ^2 }}{{\left( {1 + \xi } \right)^4 }} - \frac{{\frac{{2m_0 a^2 E}}{{\hbar ^2 }}}}{{\left( {1 + \xi } \right)^2 }}} \right)\psi  = 0.
\ee

Introduction of the following constants

\be
\label{const}
c_0  = \frac{{2m_0 a^2 E}}{{\hbar ^2 }},\quad c_2  =  c_0  + \lambda _0 ^4 a^4 ,
\ee
and some easy mathematical tricks slightly simplify (\ref{seq-pdem3}) as follows:

\be
\label{seq-pdem4}
\psi '' + \frac{2}{{1 + \xi }}\psi ' + \frac{{c_0 \left( {1 + \xi } \right)^2  + \left( {c_0  - c_2 } \right)\xi ^2 }}{{\left( {1 + \xi } \right)^4 }}\psi  = 0.
\ee

Here, $\psi ' \equiv \frac{{\mathrm{d}\psi }}{{\mathrm{d}\xi }}$ and $\psi '' \equiv \frac{{\mathrm{d}^2\psi }}{{\mathrm{d}\xi^2 }}$. 

We look for the solution of the wave function as follows:

\be
\label{psi1}
\psi  = \varphi y,
\ee
where, $\varphi \equiv \varphi \left(\xi \right)$ has the following general analytical form:

\be
\label{phi1}
\varphi  = \left( {1 + \xi } \right)^A \mathrm{e}^{\frac{B}{{1 + \xi }}} .
\ee

Here, one needs to find exactly both parameters $A$ and $B$ through the substitution of (\ref{phi1}) together with (\ref{psi1}) at (\ref{seq-pdem4}). Long straightforward computations lead to the following statement that it is possible only if

\be
\label{ab}
A=B=\varepsilon \lambda _0 ^2 a^2, \quad \varepsilon =\pm 1,
\ee
thereof, one obtains for $\varphi$

\be
\label{phi2}
\varphi  = \left( {1 + \xi } \right)^{\varepsilon \lambda _0 ^2 a^2 } \mathrm{e}^{\frac{{\varepsilon \lambda _0 ^2 a^2 }}{{1 + \xi }}} .
\ee

Now, one needs to define $\varepsilon $. We take into account that for our quantum model the following limit relations have to hold:

\[
\mathop {\lim }\limits_{\xi \to  + \infty } \varphi  = \mathop {\lim }\limits_{\xi \mathop  \to \limits   - 1^+} \varphi  = 0.
\]

This means that 

\[
\varepsilon  =  - 1.
\]

Then, the final analytical expression of $\varphi$ is as follows:

\be
\label{phi3}
\varphi  = \left( {1 + \xi } \right)^{- \lambda _0 ^2 a^2 } \mathrm{e}^{\frac{{- \lambda _0 ^2 a^2 }}{{1 + \xi }}} .
\ee

Thanks to exact expression of (\ref{phi3}), the second-order differential equation for $\psi$ (\ref{seq-pdem4}) transfers to the following second-order differential equation for $y$:

\be
\label{y1}
\left( {\xi  + 1} \right)^2 y'' + 2\left[ {\left( {1 - \lambda _0 ^2 a^2 } \right)\xi  + 1} \right]y' = \left( {\lambda _0 ^2 a^2  - c_0 } \right)y.
\ee

Further, one needs to solve Eq.(\ref{y1}) to obtain exact eigenvalues of the energy spectrum. First, it is necessary to take into account that the polynomial $\sigma \left( \xi  \right) = \left( {\xi + 1 } \right)^2$ has multiple roots. Therefore, following the method from~\cite{nikiforov1988}, it is better to apply the transition to the new variable $\tau$ as follows:

\be
\label{tau}
\tau  = \frac{1}{{\xi  + 1}}.
\ee

Substitution of (\ref{tau}) at Eq.(\ref{y1}) leads to

\be
\label{y2}
\frac{{\mathrm{d}^2 y}}{{\mathrm{d}\tau ^2 }} + 2\lambda _0^2 a^2 \frac{{1 - \tau }}{\tau }\frac{{\mathrm{d}y}}{{\mathrm{d}\tau }} + \frac{{c_0  - \lambda _0^2 a^2 }}{{\tau ^2 }}y = 0.
\ee

We look for its solution as follows:

\be
\label{y-a}
y = \tau ^\gamma  v\left( \tau  \right),
\ee
where, $\gamma$ is an arbitrary parameter, which will be defined below. As a result of the simple computations, we obtain the following equation for $v\left( \tau  \right)$:

\be
\label{v1}
\frac{{\mathrm{d}^2 v}}{{\mathrm{d}\tau ^2 }} + \frac{{\mu  - 2\lambda _0^2 a^2 \tau }}{\tau }\frac{{\mathrm{d}v}}{{\mathrm{d}\tau }} + \frac{{\rho  - \delta \tau }}{{\tau ^2 }}v = 0,
\ee
where $\mu  = 2\left( {\gamma  + \lambda _0^2 a^2 } \right)$, $\rho  = c_0  - \lambda _0^2 a^2  + \gamma \left( {\gamma  - 1} \right) + 2\gamma \lambda _0^2 a^2 $ and $\delta  = 2\gamma \lambda _0^2 a^2$.

Now we choose the constant parameter $\gamma$ in such a way that it satisfies the condition $\rho=0$. Hence we get

\be
\label{gamma}
\gamma  = \frac{{1 - 2\lambda _0^2 a^2  + e\sqrt {1 + 4\lambda _0^4 a^4  - 4c_0 } }}{2},\quad e =  \pm 1.
\ee

As a consequence of this condition, one obtains from Eq.(\ref{v1}) the following equation for $v\left(\tau \right)$:

\be
\label{v2}
\tau \frac{{\mathrm{d}^2 v}}{{\mathrm{d}\tau ^2 }} + \left( {\mu  - 2\lambda _0^2 a^2 \tau } \right)\frac{{\mathrm{d}v}}{{\mathrm{d}\tau }} - \delta v = 0.
\ee

Now, if to change the variable from $\tau$ to $z=2\lambda _0^2 a^2\tau$, then the equation above will overlap with the following hypergeometric type equation~\cite{landau1991}:

\be
\label{v3}
z\frac{{\mathrm{d}^2 v}}{{\mathrm{d}z^2 }} + \left( {\mu  - z} \right)\frac{{\mathrm{d}v}}{{\mathrm{d}z}} - \gamma v = 0.
\ee

Analytical solution to this equation in terms of the $_1F_1$ hypergeometric functions is well known:

\be
\label{v}
v = C_1 \cdot \,_1 F_1 \left( {\begin{array}{*{20}c}
   \gamma   \\
   \mu   \\
\end{array};2\lambda _0^2 a^2 \tau } \right) + C_2 \cdot\tau^{1-\mu} \,_1 F_1 \left( {\begin{array}{*{20}c}
   {\gamma  - \mu  + 1}  \\
   {2 - \mu }  \\
\end{array};2\lambda _0^2 a^2 \tau } \right).
\ee

\subsection{The states of the discrete spectrum ($E<V_\infty$)}

Condition $E<V_\infty$ requires for $\gamma$ parameter from definition (\ref{gamma}) that $e=+1$ is only possible case. Therefore, one simplifies

\be
\label{gamma1}
\gamma  = \frac{{1 - 2\lambda _0^2 a^2  + \sqrt {1 + 4\lambda _0^4 a^4  - 4c_0 } }}{2}.
\ee

Moreover, the condition for the positivity of the root expression should be satisfied, too:

\be
\label{pc}
1 + 4\lambda _0^2 a^2  - 4c_0  = \left( {1 - 2\lambda _0^2 a^2 } \right)^2  + 4\lambda _0^2 a^2  - 4c_0  > 0,
\ee
where, $\lambda _0^2 a^2-c_0<0$. Due to that $\mu  = 1 + \sqrt {1 + 4\lambda _0^4 a^4  - 4c_0 }  > 1$, hence the finiteness requirement at value $\tau=0$ for function $v\left(\tau \right)$ defined through (\ref{v}) leads to $C_2=0$. Then, we obtain an analytical expression of function $y$ (\ref{y-a}) as follows:

\[
y\left( \tau  \right) = \tau ^\gamma  \,_1 F_1 \left( {\begin{array}{*{20}c}
   \gamma   \\
   \mu   \\
\end{array};2\lambda _0^2 a^2 \tau } \right).
\]

Analytical expression of the wavefunction (\ref{psi1}) also can be written down as follows:

\be
\label{psi2}
\psi \left( \xi  \right) = C\left( {1 + \xi } \right)^{ - \gamma  - \lambda _0 ^2 a^2 } \mathrm{e}^{ - \frac{{\lambda _0 ^2 a^2 }}{{1 + \xi }}} \,_1 F_1 \left( {\begin{array}{*{20}c}
   \gamma   \\
   \mu   \\
\end{array};2\frac{{\lambda _0 ^2 a^2 }}{{1 + \xi }}} \right), \qquad C \equiv C_1.
\ee

It is finite at $\xi \to \infty$. Then, another requirement $\psi \left( -1  \right) =0$ leads to quantization of the energy through $\gamma=-n$, i.e.

\[
\sqrt {1 + 4\lambda _0 ^4 a^4  - 4c_0 }  = 2\lambda _0 ^2 a^2  - 2n - 1,\quad n = 0,1,2,3, \ldots .
\]

Positivity condition (\ref{pc}) restricts $n$ from above, i.e. due to condition $2\lambda _0 ^2 a^2  - 2n - 1>0$, there is a maximum value of $n$ as $n_{max}=N=0,1,2,3, \ldots$. Here, $N$ itself satisfies the condition

\be
\label{nr}
N < \lambda _0 ^2 a^2  - \frac{1}{2}.
\ee

Now, taking into account above listed restrictions one easily obtains for the energy spectrum that
\be
\label{eb-n}
E \equiv E_n^{QW}  = \hbar \omega \left( {n + \frac{1}{2}} \right) - \frac{{\hbar ^2 }}{{2m_0 a^2 }}n\left( {n + 1} \right),\quad n = 0,1,2, \ldots,N.
\ee

Hence, an analytical expression of the wavefunctions of the stationary states (\ref{psi2}) also slightly changes as follows:

\be
\label{psi3}
\psi _n \left( \xi  \right) = C_n \left( {1 + \xi } \right)^{n - \lambda _0 ^2 a^2 } \mathrm{e}^{ - \frac{{\lambda _0 ^2 a^2 }}{{1 + \xi }}} \,_1 F_1 \left( {\begin{array}{*{20}c}
   { - n}  \\
   {2\lambda _0 ^2 a^2  - 2n}  \\
\end{array};2\frac{{\lambda _0 ^2 a^2 }}{{1 + \xi }}} \right),\quad n = 0,1,2, \ldots,N.
\ee

Using the following well-known expression of the generalized Laguerre polynomials through $_1F_1$ hypergeometric functions~\cite{koekoek2010}

\[
L_n^\alpha  \left( x \right) = \frac{{\left( {\alpha  + 1} \right)_n }}{{n!}}\,_1 F_1 \left( {\begin{array}{*{20}c}
   { - n}  \\
   {\alpha  + 1}  \\
\end{array};x} \right),
\]
analytical expression of the wavefunctions of the stationary states (\ref{psi3}) also can be written down through these polynomials as follows:

\be
\label{psi}
\psi _n \left( \xi  \right) = C_n \left( {1 + \xi } \right)^{n - \lambda _0 ^2 a^2 } \mathrm{e}^{ - \frac{{\lambda _0 ^2 a^2 }}{{1 + \xi }}} L_n^{2\lambda _0 ^2 a^2  - 2n - 1} \left( {2\frac{{\lambda _0 ^2 a^2 }}{{1 + \xi }}} \right).
\ee

One needs to note that there exist Bessel polynomials from the Askey scheme of the orthogonal polynomials~\cite{koekoek2010}, which are expressed through both $_1F_1$ and $_2F_0$ hypergeometric functions as follows:

\begin{eqnarray}
y_n \left( x;\alpha \right) &=& \left( {n + \alpha  + 1} \right)_n \left( {\frac{x}{2}} \right)^n \,_1 F_1 \left( {\begin{array}{*{20}c}
   { - n}  \\
   { - 2n - \alpha }   \\
\end{array};\frac{2}{x}} \right) \nonumber\\ 
 &=& \,_2 F_0 \left( {\begin{array}{*{20}c}
   { - n,n + \alpha  + 1}  \\
    -   \\
\end{array}; - \frac{x}{2}} \right),\quad n = 0,1,2, \ldots ,N. \nonumber 
\end{eqnarray}

Bessel polynomials $y_n \left( {\frac{{2\xi  + 2d}}{\beta };\alpha } \right)$ satisfy the second-order differential equation of the following form~\cite{koekoek2010}

\be
\label{bessel}
\left( {\xi  + d} \right)^2 y_n ''\left( \xi  \right) + \left[ {\left( {\alpha  + 2} \right)\left( {\xi  + d} \right) + \beta } \right]y_n '\left( \xi  \right) = n\left( {n + \alpha  + 1} \right)y_n \left( \xi  \right),
\ee
where

\[
 - d < \xi  <  + \infty,\quad \alpha  <  - \left( {2N + 1} \right), \quad \beta  > 0,
\]
and

\[
y_n \left( \xi  \right) \equiv y_n \left( {\frac{{2\xi  + 2d}}{\beta };\alpha } \right).
\]

Therefore, wavefunctions (\ref{psi3}) can also be written down in terms of the Bessel polynomials, and the normalization factor $C_n$ for these wavefunctions can be easily computed from the orthogonality relation of these polynomials.

Taking into account that the following orthogonality relation holds for the Bessel polynomials $y_n \left( {\frac{{2\xi  + 2d}}{\beta };\alpha } \right)$~\cite{koekoek2010}:

\[
\int\limits_{ - d}^{ + \infty } {\left( {\xi  + d} \right)^\alpha  \mathrm{e}^{ - \frac{\beta }{{\xi  + d}}} y_m \left( {\frac{{2\xi  + 2d}}{\beta };\alpha } \right)y_n \left( {\frac{{2\xi  + 2d}}{\beta };\alpha } \right)\mathrm{d}\xi }  =  - \frac{{\beta ^{\alpha  + 1} }{\Gamma \left( { - n - \alpha } \right)n!}}{{2n + \alpha + 1}}\delta _{mn} ,
\]
one finds the orthonormalized wavefunctions (\ref{psi3}) or (\ref{psi}) in terms of the Bessel polynomials, which are as follows:

\be
\label{wf-final}
\psi _n^{QW} \left( x \right) = C_n ^{QW} \left( {\frac{x}{a} + 1} \right)^{ - \lambda _0 ^2 a^2 } \mathrm{e}^{ - \frac{{\lambda _0 ^2 a^3 }}{{x + a}}} y_n \left( {\frac{{x + a}}{{\lambda _0 ^2 a^3 }}; - 2\lambda _0 ^2 a^2 } \right),
\ee
where the orthonormalized coefficients $C_n ^{QW}$ are

\be
\label{cnqw}
C_n ^{QW}  = \left( {2\lambda _0 ^2 a^2 } \right)^{\lambda _0 ^2 a^2 } \sqrt {\frac{{2\lambda _0 ^2 a^2  - 2n - 1}}{{2\lambda _0 ^2 a^3 n!\Gamma \left( {2\lambda _0 ^2 a^2  - n} \right)}}}. 
\ee

\subsection{The states of the continuous spectrum ($E>V_\infty$)}

The spectrum of the energy values, satisfying the condition $E>V_\infty$ is continuous and spreads from $V_\infty   = \frac{{m_0 \omega ^2 a^2 }}{2}$ to $\infty$. In this case we have $c_0  > \lambda _0 ^4 a^4  + \frac{1}{4}$. Hence, the value of $\gamma$ defined via Eq.(\ref{gamma}) and consequently, values of $\mu$, $\rho$ and $\delta$ are complex. For example,

\be
\label{gamma2}
\gamma  = \frac{{1 - 2\lambda _0 ^2 a^2  + i\sqrt {4c_0  - 4\lambda _0 ^4 a^4  - 1} }}{2}.
\ee

Then, the eigenfunctions of the continuous spectrum will have the following expression:

\be
\label{psi-c}
\psi _E \left( x \right) = {\mathop{\rm const}\nolimits}  \cdot \left( {\frac{x}{a} + 1} \right)^{ - \gamma  - \lambda _0 ^2 a^2 } \mathrm{e}^{ - \frac{{\lambda _0 ^2 a^3 }}{{x + a}}} \,_1 F_1 \left( {\begin{array}{*{20}c}
   \gamma   \\
   \mu   \\
\end{array};\frac{{2\lambda _0 ^2 a^3 }}{{x + a}}} \right).
\ee

By obtaining exact expressions for the discrete energy spectrum (\ref{eb-n}) and the wavefunctions of the stationary states (\ref{wf-final}) and (\ref{psi-c}) we achieved our main goal. We solved exactly the Schr\"odinger equation (\ref{seq-pdem1}) corresponding to the potential (\ref{vx}). We are going to discuss the main properties of the discrete energy spectrum (\ref{eb-n}) and the wavefunctions of the stationary states (\ref{wf-final}) corresponding to this discrete spectrum in the final section, where we will also show that how both the energy spectrum (\ref{eb-n}) and the wavefunctions of the stationary states (\ref{wf-final}) reduce to the energy spectrum (\ref{en}) and the wavefunctions of the stationary states (\ref{wf-on}) under the limit $a \to +\infty$.

\section{Limit cases and conclusions}

Recently, \cite{jafarov2021,jafarov2022} developed the exactly-solvable model of a one-dimensional nonrelativistic canonical semiconfined quantum harmonic oscillator with a mass that varies with position. Analytical definition of the mass of that model differs from the analytical definition (\ref{m-x}) with the power of the denominator $(a+x)$ that equals to $1$. After, \cite{quesne2022} generalized the same model to the case when the power of the denominator $(a+x)$ of the position-dependent effective mass is greater than $0$, but, less than $2$. One deduces from these results that if the power of the denominator $(a+x)$ with less than $2$ still allows the quantum system to exhibit only the semiconfinement effect, but, with the power of the denominator that equals to $2$ behavior of the quantum system under the potential (\ref{vx}) changes drastically and both infinitely high wall for the negative value and finite wall for the positive value of the position appear.

Let's discuss the results obtained within this paper in more detail. This will allow us to understand the main differences of both the energy spectrum (\ref{eb-n}) and the wavefunctions of the stationary states (\ref{wf-final}) from to the energy spectrum (\ref{en}) and the wavefunctions of the stationary states (\ref{wf-on}). First of all, one needs to highlight the main difference -- it arises during the comparison of (\ref{y1}) with the equation for the Bessel polynomials (\ref{bessel}) as well as through condition (\ref{nr}), from where one obtains that

\be
\label{a-rest}
N < \lambda _0 ^2 a^2  - \frac{1}{2} \Rightarrow a > \frac{1}{{\lambda _0 }}\sqrt {N + \frac{1}{2}}.
\ee

One observes from the expression of the energy spectrum (\ref{eb-n}) that it is discrete, but, unlike the energy spectrum (\ref{en}), it is non-linear and finite, where a number of its finite levels is defined by the semiconfinement parameter $a$ via inequality (\ref{a-rest}). The origin of such behavior is due to that if our potential is confined from the negative side by the infinitely high wall, then the maximal height of the wall from the right side is $\frac{{m_0\omega ^2 a^2}}{2}$ due to definition of the potential (\ref{vx}). Due to that $\min \left( N \right) = 0$, one observes that
\be
\label{min-a}
\min \left( a \right) > \frac{1}{{\sqrt 2 \lambda _0 }}.
\ee

This condition allows us to say that our definition $a>0$ introduced at the initial stage during the definition of the potential (\ref{vx}) is not completely accurate, therefore, should be replaced by a more accurate definition (\ref{min-a}).

\begin{figure}
\centering
\resizebox{0.4\textwidth}{!}{%
  \includegraphics{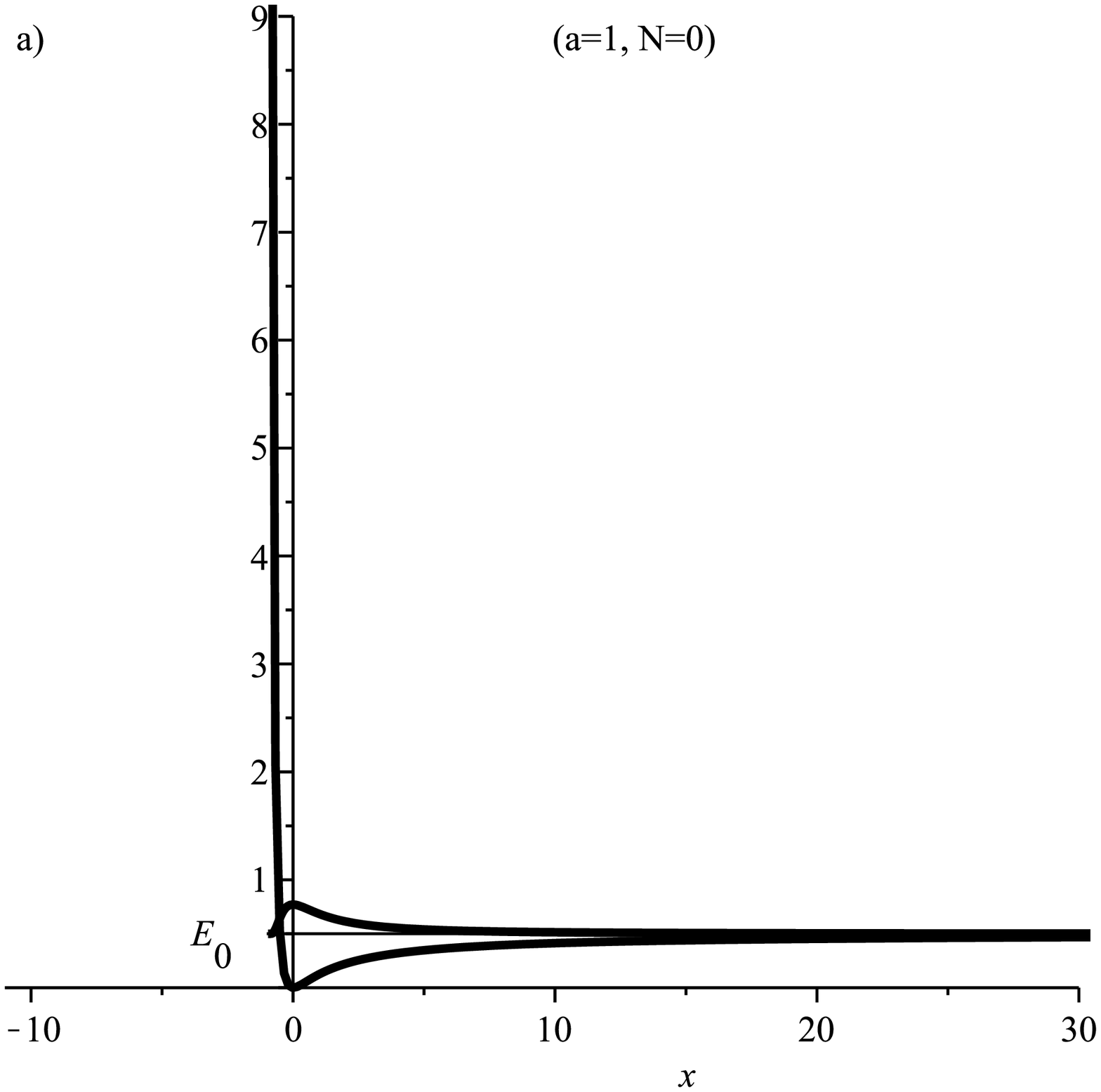}
}
\resizebox{0.4\textwidth}{!}{%
  \includegraphics{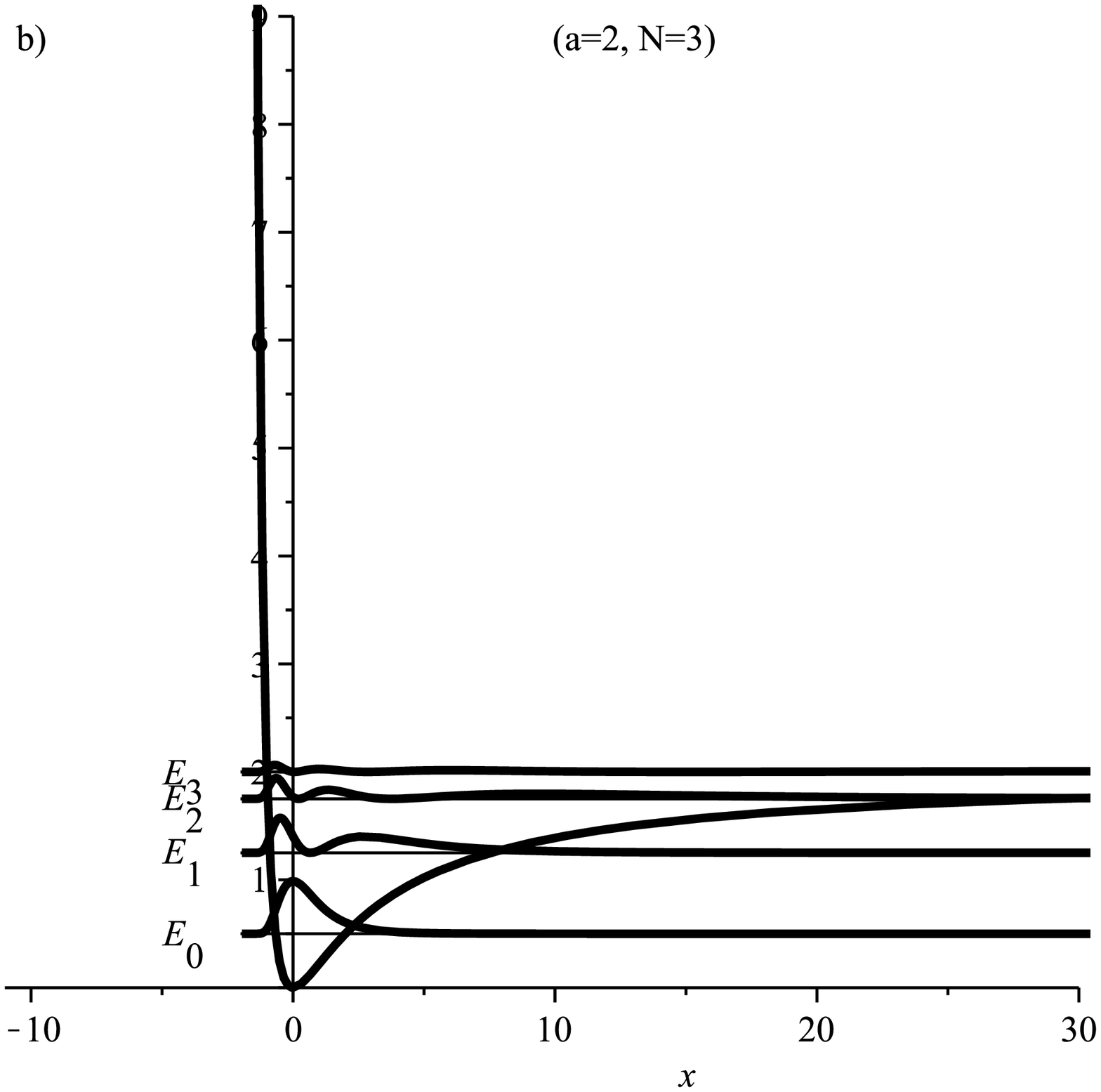}
}\\
\resizebox{0.4\textwidth}{!}{%
  \includegraphics{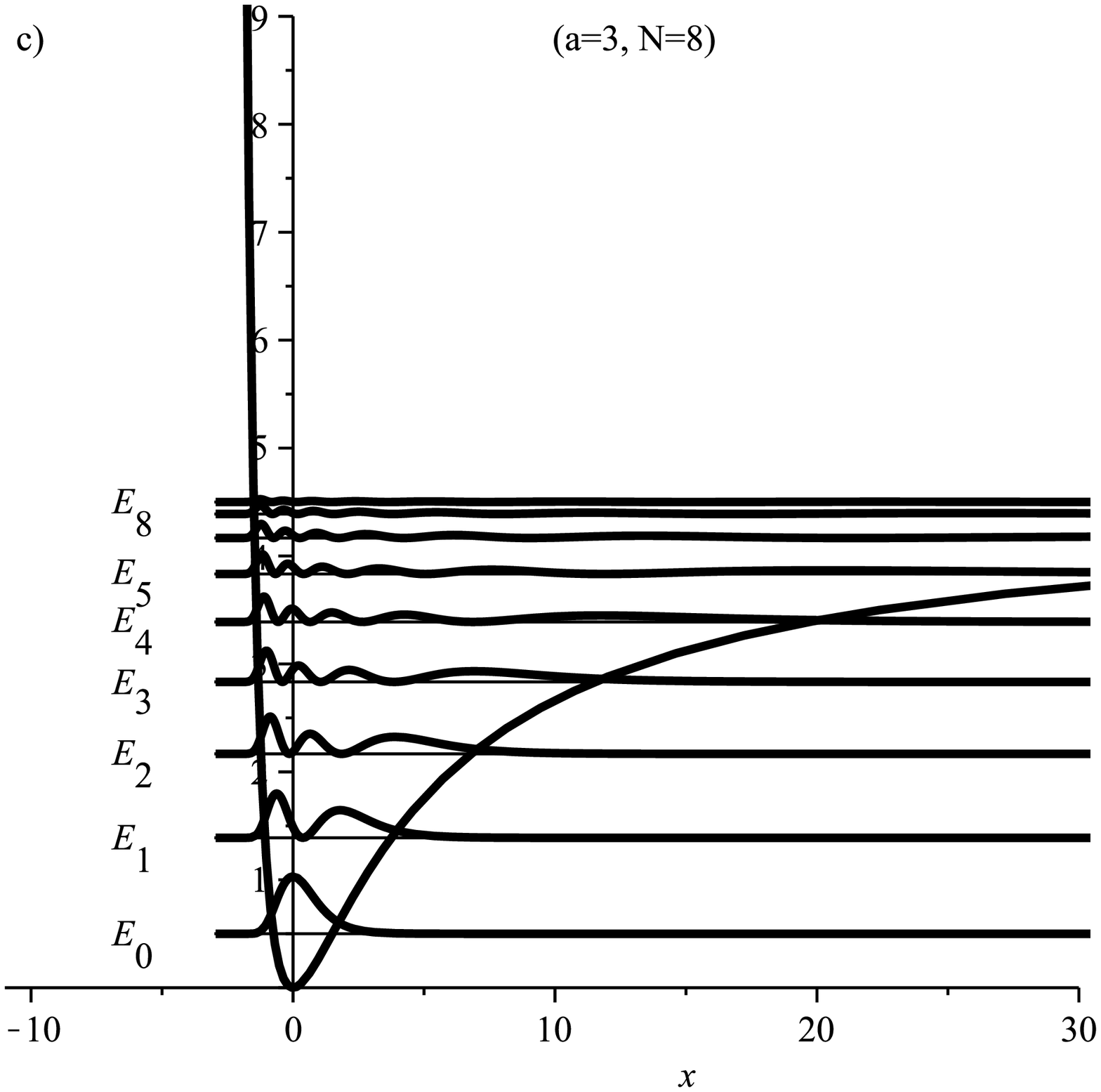}
}
\resizebox{0.4\textwidth}{!}{%
  \includegraphics{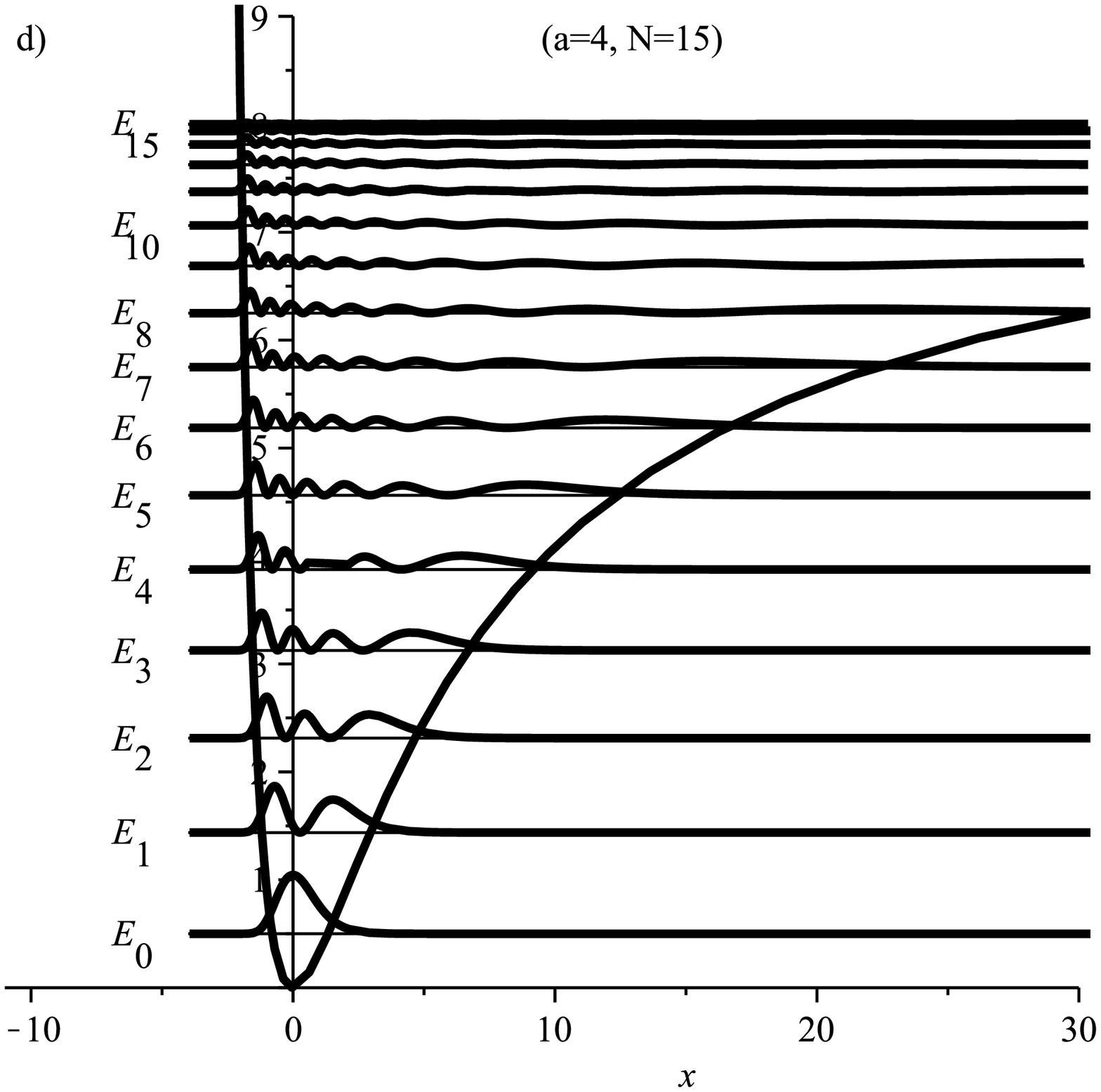}
}
\caption{The semi-infinite quantum well potential~(\ref{vx}) of the non-rectangular step-harmonic profile, the corresponding energy levels~(\ref{eb-n}) and the probability densities $\lvert {\psi _n^{QW} (x)} \rvert^2$ of the wavefunctions of the stationary states~(\ref{wf-final}) for the ground and $N$ excited states for the confinement parameter: a) $a=1$, so $N=0$; b) $a=2$, so $N=3$; c) $a=3$, so $N=8$ and d) $a=4$, so $N=15$ ($m_0=\omega=\hbar=1$, hence, $\lambda _0 = 1$).} 
\label{fig.2}
\end{figure}

In Fig.2, we depicted the semi-infinite quantum well potential~(\ref{vx}) of the non-rectangular step-harmonic profile, the corresponding energy levels~(\ref{eb-n}) and the probability densities $\lvert {\psi _n^{QW} (x)} \rvert^2$ of the wavefunctions of the stationary states~(\ref{wf-final}) for the ground and $N$ excited states, where these states correspond for different values of the confinement parameter $a$, namely, there are plots for $a=1$; $2$; $3$ and $4$. For simplicity, all computations for these plots made for values $m_0=\omega=\hbar=1$ ($\lambda _0 = 1$). One observes when $a=1$, that is a special case of inequality (\ref{min-a}) and corresponds to the plot 1(a), only the ground state exists for the quantum system under construction. This is due to that $\max\left(N\right)=0$ if $a=1$. Such a property of the quantum system under consideration is very attractive due to that from a mathematical viewpoint, the polynomial itself is one if $N=0$. Then, from a physics viewpoint there is only the ground state energy level that equals to $\frac {\hbar \omega}2$ and completely overlaps with the ground state energy level of the non-relativistic quantum harmonic oscillator, but, the wavefunction of the ground state differs from the wavefunction of the ground state of the non-relativistic quantum harmonic oscillator that exhibits the Gaussian distribution. Therefore, one observes from Fig.1a that the probability of finding the quantum system close the infinitely high wall is higher than at values $x=0$ and $x \to \frac{m_0\omega^2 a^2}2$. As a consequence of such a property, excited states appear and their number increases coherently only with the increase of the value of the confinement parameter $a$. Such a behavior can be easily observed from the plots (b)--(d) of Fig. 2. We did not present here the plot when $a \to +\infty$ completely recovers the so-called Hermite oscillator model, but, one can easily compute the limit from the energy (\ref{eb-n}) to the energy spectrum (\ref{en}). Also, it is possible to show that there is a correct limit from the wavefunctions of the stationary states (\ref{wf-final}) to the wavefunctions of the stationary states (\ref{wf-on}) under the case $a \to +\infty$.

During computation of the above-mentioned limit relations, we are going to use the following approximate ($\lvert x \rvert <  < 1$) and asymptotic ($\lvert x \rvert \to \infty$) formulae:

\begin{eqnarray}
\label{apprx}
 \sqrt {1 + x}  \approx 1 + \frac{1}{2}x - \frac{1}{8}x^2  , \\ 
\label{apprx2}
 \frac{1}{{1 + x}} \approx 1 - x + x^2  , \\ 
\label{apprx3}
 \ln \left( {1 + x} \right) \approx x - \frac{1}{2}x^2 , \\ 
\label{apprx4}
 \Gamma \left( x \right) \cong \sqrt {\frac{{2\pi }}{x}} \mathrm{e}^{x\ln x - x} .
\end{eqnarray}

As we noted, one can easily show that the following limit relation is correct:

\[
\mathop {\lim }\limits_{a \to \infty } E_n ^{QW}  = E_n  = \hbar \omega \left( {n + \frac{1}{2}} \right).
\]

Here, one needs to highlight that ground state levels ($n=0$) of both energy spectrum expressions (\ref{eb-n}) and (\ref{en}) already overlap, i.e.:

\[
E_0 ^{QW}  = E_0=\frac{1}{2}\hbar \omega.
\]

Let's rewrite the wavefunctions of the stationary states (\ref{wf-final}) as follows:

\[
\psi _n^{QW} \left( x \right) = C_n ^{QW} \varphi ^{QW} \left( x \right)y_n ^{QW} \left( x \right),
\]
where one has ($b=\lambda_0 a$ is introduced to make simpler the expression of the wavefunction)
\[
C_n ^{QW}  = \left( {2b^2 } \right)^{b^2 } \sqrt {\frac{{\lambda _0 \left( {2b^2  - 2n - 1} \right)}}{{2b^3 n!\Gamma \left( {2b^2  - n} \right)}}} ,
\]
\[
\varphi ^{QW} \left( x \right) = \left( {\frac{{\lambda _0 x}}{b} + 1} \right)^{ - b^2 } e^{ - \frac{{b^3 }}{{\lambda _0 x + b}}} ,
\]
\[
y_n ^{QW} \left( x \right) = y_n \left( {\frac{{\lambda _0 x + b}}{{b^3 }}; - 2b^2 } \right).
\]

Taking into account (\ref{apprx}) one can rewrite that

\[
C_n ^{QW}  \approx \sqrt {\frac{{\lambda _0 }}{{2bn!}}} \frac{{\left( {2b^2 } \right)^{b^2 } }}{{\sqrt {\Gamma \left( {2b^2  - n} \right)} }}.
\]

Here, taking into account Stirling's approximation (\ref{apprx4}) one has

\be
\label{cn-lim}
C_n ^{QW}  \approx \sqrt {\frac{{2^n \lambda _0 }}{{n!\sqrt \pi  }}} b^n \mathrm{e}^{b^2 } .
\ee

Now, let's explore the asymptotical behavior of $\varphi ^{QW} \left( x \right)$. Using (\ref{apprx2}) and (\ref{apprx3}) and as a result of the straightforward computations one obtains that

\be
\label{phi-lim}
\varphi ^{QW} \left( x \right) = \mathrm{e}^{ - b^2 \ln \left( {1 + \frac{{\lambda _0 x}}{b}} \right)} \mathrm{e}^{ - b^2 \left( {1 + \frac{{\lambda _0 x}}{b}} \right)^{ - 1} }  \cong \mathrm{e}^{b^2  - \frac{1}{2}\lambda _0 ^2 x^2 } .
\ee

To obtain the above expression, we did not take into account the terms of the expansions, which do not give any contribution to the limit relation $\mathop {\lim }\limits_{a \to \infty } \psi _n ^{QW} \left( x \right) = \psi _n \left( x \right)$.

Combining together Eqs.(\ref{cn-lim}) and (\ref{phi-lim}) one obtains

\be
\label{cnphi}
C_n ^{QW} \varphi ^{QW} \left( x \right) \cong \sqrt {\frac{{\lambda _0 }}{{n!\sqrt \pi  }}} \left( {\sqrt 2 b} \right)^n \mathrm{e}^{ - \frac{1}{2}\lambda _0 ^2 x^2 } .
\ee

The correct limit that reduces the Bessel polynomials directly to the Hermite polynomials is presented below as a proposition.

\begin{proposition}[Limit relation from $y_n \left( {x;\alpha} \right)$ to $H_n\left(x\right)$]
The Hermite polynomials $H_n\left(x\right)$ follow from the Bessel polynomials $y_n \left( {x;\alpha} \right)$ by setting $x\to 2/\nu+2/\nu\sqrt{2/\nu}x$ and $\alpha \to -\nu$ and then letting $\nu \to \infty$ in the following way:
\be
\label{lim-b-h}
\mathop {\lim }\limits_{\nu  \to \infty } \left( { - 1} \right)^n \left( {2\nu } \right)^{n/2} y_n \left( {\frac{2}{\nu } + \frac{2}{\nu }\sqrt {\frac{2}{\nu }} x; - \nu } \right) = H_n \left( x \right).
\ee
\end{proposition}

\begin{proof}
We prove the correctness of the above proposition using the mathematical induction technique. We start from the following Rodrigues-type formula for the Bessel polynomials~\cite{koekoek2010}:

\[
y_n \left( {x;\alpha } \right) = 2^{ - n} x^{ - \alpha } \mathrm{e}^{\frac{2}{x}} \frac{{\mathrm{d}^n }}{{\mathrm{d}x^n }}\left( {x^{2n + \alpha } \mathrm{e}^{ - \frac{2}{x}} } \right).
\]

Then, one easily finds that

\be
\label{y1y2}
y_1 \left( {x;\alpha } \right) = \frac{1}{2}\left( {2 + \alpha } \right)x + 1,\quad y_2 \left( {x;\alpha } \right) = \frac{1}{4}\left( {3 + \alpha } \right)\left( {4 + \alpha } \right)x^2  + \left( {3 + \alpha } \right)x + 1.
\ee

Now, one can check that the limit relation (\ref{lim-b-h}) is correct for both $y_1 \left( {x;\alpha } \right)$ and $y_2 \left( {x;\alpha } \right)$, i.e.

\[
 - \mathop {\lim }\limits_{\nu  \to \infty } \sqrt {2\nu } \cdot y_1 \left( {\frac{2}{\nu } + \frac{2}{\nu }\sqrt {\frac{2}{\nu }} x; - \nu } \right) = 2x = H_1 \left( x \right),
\]
\[
\mathop {\lim }\limits_{\nu  \to \infty } 2\nu \cdot y_2 \left( {\frac{2}{\nu } + \frac{2}{\nu }\sqrt {\frac{2}{\nu }} x; - \nu } \right) = 4x^2  - 2 = H_2 \left( x \right).
\]

Let's prove the correctness of the limit relation (\ref{lim-b-h}) for case $n >2$. One uses the following recurrence relation for the Bessel polynomials~\cite{koekoek2010}:

\be
\label{rr-b}
y_{n + 1} \left( {x;\alpha } \right) = A_n y_n \left( {x;\alpha } \right) + B_n y_{n - 1} \left( {x;\alpha } \right),
\ee
where,

\be
\label{an-bn}
A_n  = \frac{{\left( {2n + \alpha  + 1} \right)\left[ {2\alpha  + \left( {2n + \alpha } \right)\left( {2n + \alpha  + 2} \right)x} \right]}}{{2\left( {n + \alpha  + 1} \right)\left( {2n + \alpha } \right)}},\;B_n  = \frac{{n\left( {2n + \alpha  + 2} \right)}}{{\left( {n + \alpha  + 1} \right)\left( {2n + \alpha } \right)}}.
\ee

Now, let the limit relation (\ref{lim-b-h}) is correct for the polynomials $y_{n} \left( {z;-\nu } \right)$ and $y_{n-1} \left( {z;-\nu } \right)$ for arbitrary $n$ and $z =\frac 2\nu + \frac 2\nu \sqrt{\frac 2\nu}x$. Then, this relation also will be correct for case of the polynomial $y_{n+1} \left( {z;-\nu } \right)$. To prove the correctness of this statement, one needs to apply the substitutions $z =\frac 2\nu + \frac 2\nu \sqrt{\frac 2\nu}x$ and $\alpha = - \nu$ at the recurrence relation (\ref{rr-b}), next, multiply its both sides to the factor $\left(-1\right)^{n+1} \left( 2\nu \right)^{\frac{n+1}2}$ and then, take a limit $\nu \to \infty$ from it. We will have

\be
\label{l-bn1}
\mathop {\lim }\limits_{\nu  \to \infty } \left( { - 1} \right)^{n + 1} \left( {2\nu } \right)^{\frac{{n + 1}}{2}} y_{n + 1} \left( {\frac{2}{\nu } + \frac{2}{\nu }\sqrt {\frac{2}{\nu }} x; - \nu } \right) = \bar A_n H_n \left( x \right) + \bar B_n H_{n - 1} \left( x \right).
\ee

Here,

\be
\label{abar-bbar}
\bar A_n  =  - \mathop {\lim }\limits_{\nu  \to \infty } \sqrt {2\nu } A_n ,\;\bar B_n  = \mathop {\lim }\limits_{\nu  \to \infty } 2\nu B_n .
\ee

By substituting (\ref{an-bn}) at (\ref{abar-bbar}) one easily computes that

\[
\bar A_n  = 2x ,\;\bar B_n  =  - 2n.
\]

Then, the limit relation (\ref{l-bn1}) will be rewritten as follows:

\be
\label{l-bn2}
\mathop {\lim }\limits_{\nu  \to \infty } \left( { - 1} \right)^{n + 1} \left( {2\nu } \right)^{\frac{{n + 1}}{2}} y_{n + 1} \left( {\frac{2}{\nu } + \frac{2}{\nu }\sqrt {\frac{2}{\nu }} x; - \nu } \right) = 2xH_n \left( x \right) - 2nH_{n - 1} \left( x \right).
\ee

Now, thanks to the recurrence relations for the Hermite polynomials (\ref{rr-h}), one observes that

\[
\mathop {\lim }\limits_{\nu  \to \infty } \left( { - 1} \right)^{n + 1} \left( {2\nu } \right)^{\frac{{n + 1}}{2}} y_{n + 1} \left( {\frac{2}{\nu } + \frac{2}{\nu }\sqrt {\frac{2}{\nu }} x; - \nu } \right) = H_{n + 1} \left( x \right).
\]

This proves the proposition on the correctness of the limit relation (\ref{lim-b-h}).

\end{proof}

One needs to note that~\cite{lesky1998} studies various properties of the Bessel polynomials. The correct limit that reduces the Bessel polynomials directly to the Hermite polynomials is among these studied properties. However, the limit relation studied there was not suitable for application to the model under present construction. Therefore, we decided to introduce a new limit relation that reduces the Bessel polynomials directly to the Hermite polynomials through (\ref{lim-b-h}) and it is exactly what we need to prove the correct reduction of the wavefunction (\ref{wf-final}) to the wavefunction (\ref{wf-on}) if $a \to \infty$. Despite that Bessel polynomials belong to the family of the polynomials introduced within the finiteness property of the Jacobi-like polynomials~\cite{routh1885,romanovski1929} and there are correct limit relations from the pseudo-Jacobi to Bessel polynomials~\cite{koekoek2010} as well as from the pseudo-Jacobi to Hermite polynomials~\cite{jafarov2021b,nagiyev2022}, but limit relation that reduces the Bessel polynomials directly to the Hermite polynomials through (\ref{lim-b-h}) completes all possible relations between these polynomials and it is also possible to maintain the direct relation between it and the limit from~\cite{lesky1998}.

Now, combining together (\ref{cnphi}) and limit relation (\ref{lim-b-h}) between the Bessel and Hermite polynomials, one observes that the wavefunction (\ref{wf-final}) also correctly reduces to the wavefunction (\ref{wf-on}) if $a \to \infty$, i.e.

\[
\mathop {\lim }\limits_{a \to \infty } \psi _n ^{QW} \left( x \right) = \psi _n \left( x \right).
\]

The wavefunctions of the continuous spectrum (\ref{psi-c}) simply vanish under the limit $a \to \infty$. One can slightly change Eq.(\ref{gamma2}) in terms of $\mu$ as follows:

\be
\label{gamma3}
\mu = 2\left( {\gamma  + \lambda _0 ^2 a^2 } \right) = 1 + i\sqrt {4c_0  - 4\lambda _0 ^4 a^4  - 1} .
\ee

Next, the condition $c_0  > \lambda _0 ^4 a^4  + \frac{1}{4}$ leads to statement that

\[
4c_0  - 4\lambda _0 ^4 a^4  - 1 = q^2  > 0,
\]
i.e., parameter $q$ introduced above is positive and finite. Therefore, it does not depend on parameter $a$. This means that $\mu = 1 + iq$ is also finite and does not depend on parameter $a$. Now, taking into account well-known approximation for the gamma functions

\[
\frac{{\Gamma \left( {z + \alpha } \right)}}{{\Gamma \left( {z + \beta } \right)}} \cong z^{\alpha  - \beta } ,
\]
one observes that

\be
\label{lim1}
\mathop {\lim }\limits_{a \to \infty } \left( {\frac{x}{a} + 1} \right)^{ - \gamma  - \lambda _0 ^2 a^2 }  = \mathop {\lim }\limits_{a \to \infty } \left( {\frac{x}{a} + 1} \right)^{ - \frac{\mu }{2}}  = 1,
\ee
and
\be
\label{lim2}
\mathop {\lim }\limits_{a \to \infty } \mathrm{e}^{ - \frac{{\lambda _0 ^2 a^3 }}{{x + a}}}  \to \mathrm{e}^{ - \lambda _0 ^2 a^2 } , \quad \mathop {\lim }\limits_{a \to \infty } {\kern 1pt} _1 F_1 \left( {\begin{array}{*{20}c}
   \gamma   \\
   \mu   \\
\end{array};\frac{{2\lambda _0 ^2 a^3 }}{{x + a}}} \right) \to {\kern 1pt} _0 F_1 \left( {\begin{array}{*{20}c}
    -   \\
   \mu   \\
\end{array}; - 2\lambda _0 ^4 a^4 } \right).
\ee

Hypergeometric functions $_0F_1$ can be expressed via Bessel functions of first kind $J_\alpha \left(x \right)$ as follows~\cite{prudnikov2002}:

\be
\label{1f1}
_0 F_1 \left( {\begin{array}{*{20}c}
    -   \\
   \mu   \\
\end{array}; - 2\lambda _0 ^4 a^4 } \right) = \Gamma \left( \mu  \right)\left( {2\lambda _0 ^4 a^4 } \right)^{\frac{{1 - \mu }}{{2\lambda _0 ^4 a^4 }}} J_{\mu  - 1} \left( {2\sqrt 2 \lambda _0 ^2 a^2 } \right).
\ee

The following approximation for Bessel functions of first kind $J_\alpha \left(x \right)$ at $\lvert x \rvert \to \infty$ is well known, too:

\be
\label{bes2}
J_\alpha  \left( x \right) \cong \sqrt {\frac{2}{{\pi x}}} .
\ee

Now, combining together (\ref{lim1})-(\ref{bes2}), one observes that

\[
\mathop {\lim }\limits_{a \to \infty } \psi _E \left( x \right) = 0.
\]

In the end, we would like to discuss briefly the factorization possibilities of non-relativistic Hamiltonian (\ref{h}) being the sum of BenDaniel--Duke kinetic energy operator (\ref{h0-pdem}) and the potential (\ref{vx}). In principle, due to that one replaces constant mass with the mass changing with position, then, the method for the factorization is different than the method applied for the ordinary quantum harmonic oscillator, which factorization leads to obtaining its creation and annihilation operators $\hat a^+$ and $\hat a^-$ (\ref{aa+}). The method that can be successfully applied here is already known~\cite{dabrowka1988,cooper1995,plastino1999,gonul2002,dong2007,amir2016}. One writes down the ground state wavefunction extracted from (\ref{wf-final}) as follows:

\be
\label{psi-0}
\psi _0^{QW} \left( x \right) = C_0^{QW} \left( {\frac{x}{a} + 1} \right)^{ - \lambda _0 ^2 a^2 } \mathrm{e}^{ - \frac{{\lambda _0 ^2 a^3 }}{{x + a}}} .
\ee

Next, for simplicity one needs to introduce $\alpha_0\left( x \right)$ and $\rho\left( x \right)$ functions as follows:

\[
\alpha _0 \left( x \right) = \frac{{\psi _0 '\left( x \right)}}{{\psi _0 \left( x \right)}} =  - \frac{{\lambda _0 ^2 a^2 }}{{x + a}} + \frac{{\lambda _0 ^2 a^3 }}{{\left( {x + a} \right)^2 }},\quad \rho \left( x \right) = \frac{{\hbar ^2 }}{{2M}}.
\]

Then, one can factorize Hamiltonian (\ref{h}) being sum of BenDaniel--Duke kinetic energy operator (\ref{h0-pdem}) and the potential (\ref{vx}) in terms of operators $\hat A^+$ and $\hat A^-$ as follows:

\[
\hat H = \hbar \omega \left( {\hat A^ +  \hat A^ -   + \frac{1}{2}} \right),
\]
where,
\begin{eqnarray}
 \hat A^ +   =  - \left( {\frac{\mathrm{d}}{{\mathrm{d}x}} + \alpha _0 \left( x \right)} \right)\sqrt {\frac{\rho \left( x \right)}{\hbar\omega}} , \\ 
 \hat A^ -   = \sqrt {\frac{\rho \left( x \right)}{\hbar \omega}} \left( {\frac{\mathrm{d}}{{\mathrm{d}x}} - \alpha _0 \left( x \right)} \right). \nonumber
 \end{eqnarray}

One can easily prove that the following equation is correct:

\[
\hat A^ -  \psi _0^{QW} \left( x \right) = 0,
\]
as well as the following two limit relations also hold:

\[
\mathop {\lim }\limits_{a \to  + \infty } \hat A^ \pm   = \hat a^ \pm  .
\]

Only the difference here is that both creation and annihilation operators $\hat a^\pm$ together with quantum harmonic oscillator Hamiltonian generate closed Heisenberg--Weyl algebra and it is dynamical symmetry algebra of the non-relativistic oscillator model described through Hermite polynomials. However, here it is impossible to write down the dynamical symmetry algebra of the model under study by a simpler way just employing $\hat A^\pm$ operators. Therefore, being focused on the analytical solution of the model under construction, we leave the problem related to its dynamical symmetry algebra for our future studies.

We constructed an exactly-solvable model of the non-relativistic harmonic oscillator with a position-dependent effective mass. It behaves itself as a semi-infinite quantum well of the non-rectangular profile and its profile is a step-harmonic potential, but with smooth transition between the step and harmonic oscillator potentials as a consequence of the certain analytical dependence of the effective mass from the position and semiconfinement parameter $a$. We were able to solve the Schr\"odinger equation corresponding to the model under construction and obtained that its wavefunctions of the stationary states are expressed through the $_1F_1$ hypergeometric functions if the spectrum is continuous and through the Bessel polynomials if the spectrum is discrete. We also found that the discrete energy spectrum of the model under construction is non-equidistant and finite, but, at the limit, when the parameter $a$ goes to infinity, both wavefunctions, and the energy spectrum correctly recover wavefunctions and the energy spectrum of the usual non-relativistic harmonic oscillator with a constant effective mass.

We are not aware that a similar exactly-solvable model with the wavefunctions expressed through the Bessel polynomials has been before constructed. Therefore, we believe that the exact-solubility advantage of the present model will make it definitely useful in the description of the various phenomena in physics.

\section*{Acknowldgement}

The authors would like to thank the anonymous reviewer for the valuable comments and suggestions, which substantially increased the quality of the paper.

\section*{Conflict of interest}

The authors declare that they have no conflicts of interest.

\section*{Data availability}

The paper has no associated data, however, any additional information regarding depicted figures generated from the computations done above is always available from the corresponding author upon reasonable request.


\begin{thebibliography}{00}

\bibitem{miller1984a} R.C.~Miller, A.C.~Gossard, D.A.~Kleinman and O.~Munteanu, Parabolic quantum wells with the $GaAs-Al_xGa_{1-x}As$ system, Phys. Rev. B, \textbf{29} 3740–-3743 (1984).

\bibitem{miller1984b} R.C.~Miller, D.A.~Kleinman and A.C.~Gossard, Energy-gap discontinuities and effective masses for $GaAs-Al_xGa_{1-x}As$ quantum wells, Phys. Rev. B, \textbf{29} 7085–-7087 (1984).

\bibitem{miller1985} R.C.~Miller, A.C.~Gossard and D.A.~Kleinman, Band offsets from two special $GaAs-Al_xGa_{1-x}As$ quantum well structures, Phys. Rev. B, \textbf{32} 5443–-5446 (1985).

\bibitem{gossard1986} A.C.~Gossard, R.C.~Miller and W.~Wiegmann, MBE growth and energy levels of quantum wells with special shapes, Surf. Sci., \textbf{174} 131–-135 (1986).

\bibitem{rizzi2010} L.~Rizzi, O.F.~Piattella, S.L.~Cacciatori and V.~Gorini, The step-harmonic potential, Am. J. Phys., \textbf{78} 842–-850 (2010).

\bibitem{amthong2014} A.~Amthong, WKB approximation for abruptly varying potential wells, Eur. J. Phys., \textbf{35} 065009 (2014).

\bibitem{morris2015} J.R.~Morris, New scenarios for classical and quantum mechanical systems with position-dependent mass, Quantum Stud.: Math. Found., \textbf{2} 359--370 (2015).

\bibitem{morris2017} J.R.~Morris, Short note: Hamiltonian for a particle with position-dependent mass, Quantum Stud.: Math. Found., \textbf{4} 295--299 (2017).

\bibitem{koekoek2010} R.~Koekoek, P.A.~Lesky and R.F.~Swarttouw, \textit{Hypergeometric Orthogonal Polynomials and Their $q$-Analogues}, Springer, Berlin 2010.

\bibitem{dirac1927} P.A.M.~Dirac, The quantum theory of the emission and absorption of radiation, Proc. R. Soc. A, \textbf{114} 243–-265 (1927).

\bibitem{infeld1951} L.~Infeld and T.E.~Hull, The factorization method, Rev. Mod. Phys., \textbf{23} 21–-68 (1951).

\bibitem{messiah1966} A.~Messiah, \textit{Quantum Mechanics (Vol.~I)}, Wiley, North Holland 1966.

\bibitem{bendaniel1966} D.J.~BenDaniel and C.B.~Duke, Space-charge effects on electron tunneling, Phys. Rev., \textbf{152} 683–-692 (1966).

\bibitem{harrison1961} W.A.~Harrison, Tunneling from an independent-particle point of view, Phys. Rev., \textbf{123} 85–-89 (1961).

\bibitem{giaever1960a} I.~Giaever, Energy gap in superconductors measured by electron tunneling, Phys. Rev. Lett., \textbf{5} 147–-148 (1960).

\bibitem{giaever1960b} I.~Giaever, Electron tunneling between two superconductors, Phys. Rev. Lett., \textbf{5} 464–-466 (1960).

\bibitem{gora1969} T.~Gora and F.~Williams, Theory of electronic states and transport in graded mixed semiconductors, Phys. Rev., \textbf{177} 1179–-1182 (1969).

\bibitem{zhu1983} Q.-G.~Zhu and H.~Kroemer, Interface connection rules for effective-mass wave functions at an abrupt heterojunction between two different semiconductors, Phys. Rev. B, \textbf{27} 3519–-3527 (1983).

\bibitem{vonroos1983} O.~von~Roos, Position-dependent effective masses in semiconductor theory, Phys. Rev. B, \textbf{27} 7547–-7551 (1983).

\bibitem{morrow1984} R.A.~Morrow and K.R.~Brownstein, Model effective-mass Hamiltonians for abrupt heterojunctions and the associated wave-function-matching conditions, Phys. Rev. B, \textbf{30} 678–-680 (1984).

\bibitem{li1993} T.L.~Li and K.J.~Kuhn, Band-offset ratio dependence on the effective-mass Hamiltonian based on a modified profile of the $GaAs$ - $Al_xGa_{1-x}As$ quantum well, Phys. Rev. B, \textbf{47} 12760–-12770 (1993).

\bibitem{li2009} J.~Li and M.~Ostoja-Starzewski, Fractal solids, product measures and fractional wave equations, Proc. R. Soc. A, \textbf{465} 2521–-2536 (2009).

\bibitem{lima2012} J.R.F.~Lima, M.~Vieira, C.~Furtado, F.~Moraes and C.~Filgueiras, Yet another position-dependent mass quantum model, J. Math. Phys., \textbf{53} 072101 (2012).

\bibitem{nobre2015} F.D.~Nobre, M.A.~Rego-Monteiro, Non-Hermitian PT Symmetric Hamiltonian with position-dependent masses: associated Schr\"odinger equation and finite-norm solutions, Braz. J. Phys., \textbf{45}, 79--88 (2015).

\bibitem{mustafa2019} O.~Mustafa, Position-dependent mass momentum operator and minimal coupling: point canonical transformation and isospectrality, Eur. Phys. J. Plus, \textbf{134}, 228 (2019).

\bibitem{elnabulsi2020} R.A.~El-Nabulsi, On a new fractional uncertainty relation and its implications in quantum mechanics and molecular physics, Proc. R. Soc. A, \textbf{476} 20190729 (2020).

\bibitem{kolesnikov1999} A.V.~Kolesnikov and A.P.~Silin, Quantum mechanics with coordinate-dependent mass, Phys. Rev. B, \textbf{59}, 7596--7599 (1999).

\bibitem{nikiforov1988} A.F.~Nikiforov, and V.B.~Uvarov, \textit{Special Functions of Mathematical Physics: A Unified Introduction with Applications}, Birkh\"auser, Basel 1988.

\bibitem{landau1991} L.D.~Landau and E.M.~Lifshitz, \textit{Quantum Mechanics: Non-Relativistic Theory}, Pergamon Press, Oxford 1991.

\bibitem{jafarov2021} E.I.~Jafarov and J.~Van~der~Jeugt, Exact solution of the semiconfined harmonic oscillator model with a position-dependent effective mass, Eur. Phys. J. Plus, {\bf 136} 758 (2021).

\bibitem{jafarov2022} E.I.~Jafarov and J.~Van~der~Jeugt, Exact solution of the semiconfined harmonic oscillator model with a position-dependent effective mass in an external homogeneous field, Pramana -- J. Phys., \textbf{96}, 35 (2022).

\bibitem{quesne2022} C.~Quesne, Generalized semiconfined harmonic oscillator model with a position-dependent effective mass, Eur. Phys. J. Plus, \textbf{137}, 225 (2022).

\bibitem{lesky1998} P.A.~Lesky, Einordnung der Polynome von Romanovski-Bessel in das Askey-Tableau, Z. Angew. Math. Mech., \textbf{78}, 646--648 (1998).

\bibitem{routh1885} E.J.~Routh, On  some  properties  of  certain  solutions  of  a  differential  equation  of  the second order, Proc. Lond. Math. Soc., \textbf{16}, 245--261 (1885).

\bibitem{romanovski1929} V.I.~Romanovski, Sur quelques classes nouvelles de polynomes orthogonaux, C.~R.~Acad.~Sci.~Paris, \textbf{188}, 1023--1025 (1929).

\bibitem{jafarov2021b} E.I.~Jafarov, A.M.~Mammadova and J.~Van~der~Jeugt, On the direct limit from pseudo Jacobi polynomials to Hermite polynomials, Mathematics, \textbf{9}, 88 (2021).

\bibitem{nagiyev2022} S.M.~Nagiyev, On two direct limits relating pseudo-Jacobi polynomials to Hermite polynomials and the pseudo-Jacobi oscillator in a homogeneous gravitational field, Theor. Math. Phys., \textbf{210}, 121--134 (2022).

\bibitem{prudnikov2002} A.P.~Prudnikov, Y.A.~Brychkov and O.I.~Marichev, \textit{Integrals and Series: vol.3 -- More Special Functions}, Taylor and Francis, London 2002.

\bibitem{dabrowka1988} J.W.~Dabrowska, A.~Khare and U.P.~Sukhatme, Explicit wavefunctions for shape-invariant potentials by operator techniques, J. Phys. A: Math. Gen., \textbf{21}, L195--L200 (1988).

\bibitem{cooper1995} F.~Cooper, A.~Khare and U.~Sukhatme, Supersymmetry and quantum mechanics, Phys. Rep., \textbf{251}, 267--385 (1995).

\bibitem{plastino1999} A.R.~Plastino, A.~Rigo, M.~Casas, F.~Garcias and A. Plastino, Supersymmetric approach to quantum systems with position-dependent effective mass, Phys. Rev. A, \textbf{60}, 4318--4325 (1999).

\bibitem{gonul2002} B.~G\"on\"ul, B.~G\"on\"ul, D.~Tutcu and O.~\"Ozer, Supersymmetric approach to exactly solvable systems with position-dependent effective masses, Mod. Phys. Lett. A, \textbf{17}, 2057--2066 (2002).

\bibitem{dong2007} S.-H.~Dong, J.J.~Pe\~ na, C.~Pachego-Garc\'{i}a and J.~Garc\'{i}a-Ravelo, Algebraic approach to the position-dependent mass Schr\"odinger equation for a singular oscillator, Mod. Phys. Lett. A, \textbf{22}, 1039--1045 (2007).

\bibitem{amir2016} N.~Amir and S.~Iqbal, Algebraic solutions of shape-invariant position-dependent effective mass systems, J. Math. Phys., \textbf{57}, 062105 (2016).

\end{thebibliography}
\end{document}